\documentclass[11pt]{article}
\usepackage[margin=1in]{geometry}
\usepackage{amsmath,amssymb,amsthm}
\usepackage{booktabs,array}
\usepackage{microtype}
\usepackage{placeins}
\usepackage{url}
\usepackage[hidelinks]{hyperref}
\hypersetup{
  pdftitle={A Lower Bound of 21 for 3 by 3 Matrix Multiplication over F2},
  pdfauthor={Chengu Wang}
}

\newtheorem{theorem}{Theorem}[section]
\newtheorem{lemma}[theorem]{Lemma}
\newtheorem{proposition}[theorem]{Proposition}
\newtheorem{corollary}[theorem]{Corollary}
\theoremstyle{remark}

\newcommand{\F}{\mathbb F}
\newcommand{\R}{\mathbf R}
\newcommand{\mmt}[3]{\langle #1,#2,#3\rangle}
\newcommand{\rank}{\operatorname{rank}}
\newcommand{\spn}{\operatorname{span}}
\newcommand{\GL}{\mathrm{GL}}
\newcommand{\Stab}{\operatorname{Stab}}
\newcommand{\PG}{\mathrm{PG}}
\newcommand{\im}{\operatorname{im}}
\newcommand{\T}{\mathsf T}
\title{A Lower Bound of 21 for\\
$3\times3$ Matrix Multiplication over $\mathbb{F}_2$}
\author{Chengu Wang\\\texttt{wangchengu@gmail.com}}
\date{}

\newcommand{\BinaryFlattenCount}{14}
\newcommand{\BinaryDegenerateCount}{234}
\newcommand{\BinaryForcedCount}{14}
\newcommand{\BinaryBacktrackingCount}{207}
\newcommand{\BinaryRankOneSpanCount}{27}
\newcommand{\TernaryFlattenCount}{11}
\newcommand{\TernaryDegenerateCount}{8}
\newcommand{\TernaryForcedCount}{2}
\newcommand{\TernaryBacktrackingCount}{4}
\newcommand{\TernaryRankOneSpanCount}{6}

\begin{document}
\maketitle

\begin{abstract}
  We prove that $3\times3$ matrix multiplication over $\mathbb{F}_2$ has
  bilinear complexity at least $21$, improving the previous lower bound of
  $20$. Lower bounds for restrictions of the first input constrain how many
  first factors of a decomposition can lie in each subspace. For a hypothetical
  $20$-term decomposition, the strengthened restriction bounds force all first
  factors of matrix rank at least two into a single coset of a three-dimensional
  rank-one subspace. An exhaustive computation finds no admissible first-factor
  profile with $20$ terms. We also prove that $2\times3$ by $3\times3$ matrix
  multiplication over $\mathbb{F}_3$ has rank exactly $15$: three profiles with
  $14$ terms survive the corresponding computation up to symmetry, and short
  restriction arguments exclude them. We combine Wang's automated
  framework for tensor-rank lower bounds with D'Ambrosio's capacity-and-profile
  strategy. Our computational contributions are rank-one-span searches that
  strengthen the subspace lower-bound table and a direct, symmetry-reduced
  profile enumerator that enforces all subspace capacities simultaneously.
\end{abstract}

\section{Introduction}\label{sec:intro}

Write $\R_{\F}(\mmt{n_0}{n_1}{n_2})$ for the tensor rank over $\F$ of
multiplication of an $n_0\times n_1$ matrix by an $n_1\times n_2$ matrix.
Equivalently, it is the least number of nonscalar multiplications in a
bilinear algorithm for this map. Determining this number remains difficult
even for small formats: $3\times3$ matrix multiplication is the smallest
square case whose rank is still unknown.

\subsection{Previous work}
For $2\times2$ matrix multiplication,
Strassen~\cite{strassen1969gaussian} gave an algorithm with seven
products and integer coefficients, and
Winograd~\cite{winograd1971multiplication} proved that fewer products
cannot suffice over $\mathbb Q$. Hopcroft and
Kerr~\cite{hopcroft1971minimizing} established the lower bounds
$\R_{\F_2}(\mmt{2}{2}{2})\ge7$ and
$\R_{\F_2}(\mmt{2}{3}{3})\ge15$, among others. Their
$15$-multiplication algorithm for $\mmt{2}{3}{3}$ is valid over every
field, so the latter lower bound determines its rank over $\F_2$.

Further lower bounds exploit the structure of small fields.
Ja'Ja' and Takche~\cite{ja1985improved} showed that
$\R_{\F_2}(\mmt{3}{2}{n})\ge\lceil105n/23\rceil$.
For square matrices, Bshouty~\cite{bshouty1989lower} obtained the
asymptotic bound
$\R_{\F_2}(\mmt{n}{n}{n})\ge\tfrac52n^2-o(n^2)$.
Shpilka~\cite{shpilka2001lower} further strengthened asymptotic lower
bounds over small fields.

Bl\"aser obtained several bounds that hold over arbitrary fields.
For $n\ge l\ge2$, his rectangular bound is
$\R_{\F}(\mmt{l}{m}{n})\ge lm+mn+l-m+n-3$~\cite{blaser1999lower}.
Although proved over algebraically closed fields, it also holds over
any field because extending the ground field cannot increase tensor
rank. In particular, it gives $\R_{\F}(\mmt{2}{3}{3})\ge14$.
He also proved
$\R_{\F}(\mmt{n}{n}{n})\ge\tfrac52n^2-3n$ for $n\ge3$ over
arbitrary fields~\cite{blaser1999fivehalf}, and later
$\R_{\F}(\mmt{n}{m}{n})\ge2mn+2n-m-2$ for
$m\ge n\ge3$~\cite{blaser2003complexity}. The last inequality gives
$19$ for $\mmt{3}{3}{3}$. On the upper-bound side,
Laderman~\cite{laderman1976noncommutative} gave a $23$-product
algorithm over the integers, hence over every field. These bounds
left the $3\times3$ rank between $19$ and $23$ for more than two decades.

More recently, Wang's automated finite-field
framework~\cite{wang2026automated} proved the binary lower bounds
$20$, $25$, and $29$ for $\mmt{3}{3}{3}$, $\mmt{3}{3}{4}$, and
$\mmt{3}{4}{4}$, respectively.
Starting from the framework's bound $19$ for $\mmt{2}{3}{4}$,
D'Ambrosio~\cite{dambrosio2026exact} established its exact rank
$20$ over $\F_2$ by classifying possible first-factor supports and
excluding them through restrictions. This strategy is the starting
point for the present paper.

\subsection{Our results}
Our main result raises the lower bound for $3\times3$ matrix
multiplication over $\F_2$ from $20$ to $21$.
\begin{theorem}\label{thm:binary}
  $\R_{\F_2}(\mmt{3}{3}{3})\ge21$.
\end{theorem}
Together with the known upper bound, this gives
$21\le\R_{\F_2}(\mmt{3}{3}{3})\le23$.

Shortly after the first version of this paper was posted,
Tahir~\cite{tahir2026rank21} obtained the same bound independently by a
different route. That review draft also starts from the capacity
inequalities implied by the restriction table of Wang's
framework~\cite{wang2026automated}, with five table entries strengthened
by one, and states the bound conditionally on $110$ restriction bounds
that its repository replays from the published certificates and from its
own evidence. It excludes a $20$-term decomposition by forcing the first
factors to be distinct singular matrices, reducing to seven normal forms
by symmetry, and refuting each with a branch-and-bound counting tree
carrying exact rational dual certificates. It uses neither the
common-block reduction nor the profile enumeration of the present paper,
so the two arguments corroborate each other.

Our second result determines the rank of $2\times3$ by $3\times3$
matrix multiplication over $\F_3$.
\begin{theorem}\label{thm:ternary}
  $\R_{\F_3}(\mmt{2}{3}{3})=15$.
\end{theorem}
This closes the gap between the previous lower bound $14$ and the
known upper bound $15$ over $\F_3$. Together with the equality over $\F_2$, it
supports the conjecture that $\R_{\F}(\mmt{2}{3}{3})=15$ for every
field $\F$.

\paragraph{The method.}
Let $T$ be the multiplication tensor and $A$ its first input space.
A \emph{subspace lower-bound table} assigns to every subspace $S\subseteq A$
a lower bound $L(S)\le\R(T_S)$, where $T_S$ is multiplication with the
first input restricted to $S$. Thus the table is a collection of rank
lower bounds, stored with one entry per symmetry orbit of subspaces.
We compute it using Wang's framework~\cite{wang2026automated},
strengthened by rank-one-span searches and deeper backtracking.
In an $r$-term decomposition, restriction to $S$ removes the terms whose
first factors vanish on $S$, so at most $r-L(S)$ terms can do so.
These limits are the \emph{capacities}. Following
D'Ambrosio~\cite{dambrosio2026exact}, we enumerate first-factor lists,
or \emph{profiles}, satisfying all capacities simultaneously.
For $\mmt{3}{3}{3}$ over $\F_2$, no profile survives at $r=20$:
the first-factor conditions alone exclude a decomposition, without
searching for its second and third factors. For $\mmt{2}{3}{3}$ over
$\F_3$, three profiles remain up to symmetry at $r=14$. Comparing overlapping
restrictions with six or seven active terms and flattening rank six
excludes all three.

\paragraph{What is inherited, and what is added.}
The dynamic programming over symmetry orbits of restriction spaces,
using flattening, degenerate reduction, forced products, and substitution
with backtracking, comes from Wang's framework~\cite{wang2026automated}.
The capacity-and-profile strategy and the tight and one-excess
restriction arguments follow D'Ambrosio~\cite{dambrosio2026exact}.

Beyond the two bounds, our additions are the binary common-block
reduction and two computational tools: a rank-one-span search for
stronger restriction bounds and a direct, symmetry-reduced profile
enumerator. The characterization of rank via the slice space underlying
the first tool is classical~\cite[Prop.~14.45]{burgisser1997algebraic};
the contribution is the finite search, including a family-based
enumeration for the larger binary quotients.

\paragraph{Code and certificates.}
The code, certificates, and commands for reproducing the computations
accompany this paper at
\url{https://github.com/wcgbg/matrix-multiplication-n333r21f2}.

\paragraph{Reading the paper.}
The next section explains the binary proof without the implementation
details. Section~\ref{sec:computations} describes its two computations,
Section~\ref{sec:ternary} proves the ternary result, and
Section~\ref{sec:verification} discusses verification and limitations.
The appendices contain the technical details:
Appendix~\ref{sec:profile-completeness} proves the completeness of profile
enumeration, Appendix~\ref{sec:family} establishes family-search coverage
and the safety of its pruning rules, and Appendix~\ref{sec:data} presents
the tables and the outer cases.

\section{Proof overview: \texorpdfstring{$\R_{\F_2}(\mmt{3}{3}{3})\ge21$}{Tensor rank over F2 of <3,3,3> is at least 21}}\label{sec:binary}

\subsection{Restrictions constrain the first factors}

For matrix multiplication let
\[
  A=\F^{n_0\times n_1},\qquad B=\F^{n_1\times n_2},\qquad
  C=\F^{n_0\times n_2}.
\]
The multiplication tensor $T\in A^*\otimes B^*\otimes C$ is
$\sum_{i,j,k}x_{ij}\otimes y_{jk}\otimes e_{ik}$, where $x_{ij}$ and
$y_{jk}$ are coordinate functionals and $e_{ik}$ are output matrix units.
We identify $U\in A^*$ with its coefficient matrix, so
$U(X)=\sum_{i,j}U_{ij}X_{ij}$. In particular, matrix rank of $U$ and
tensor rank of $T$ are different notions.

For a subspace $S\subseteq A$, let $T_S$ be the restriction of the first
input to $S$. Suppose that a table gives lower bounds
\[
  L(S)\le\R_{\F}(T_S)\qquad(S\subseteq A).
\]
For a proposed number of terms $r$ and a subspace $E\subseteq A^*$, put
\[
  c_r(E)=r-L(E^\perp),\qquad
  E^\perp=\{X\in A:U(X)=0\text{ for all }U\in E\}.
\]
We call $c_r(E)$ its capacity. The table itself is always indexed by the
restriction space $S$, not by its annihilator.

\begin{lemma}[capacity inequality, after D'Ambrosio~{\cite[Section~4]{dambrosio2026exact}}]\label{lem:capacity}
  If $T=\sum_{s=1}^r U_s\otimes V_s\otimes W_s$, and $m(E)$ counts the
  indices $s$ with $U_s\in E$, including multiplicities, then
  \begin{equation}\label{eq:capacity}
    m(E)\le c_r(E)=r-L(E^\perp)\qquad(E\subseteq A^*).
  \end{equation}
\end{lemma}
\begin{proof}
  Restricting the first input to $E^\perp$ kills those $m(E)$ terms. Hence
  $L(E^\perp)\le\R(T_{E^\perp})\le r-m(E)$.
\end{proof}

Following D'Ambrosio's profile formulation~\cite[Section~4]{dambrosio2026exact},
we call a multiset of $r$ projective first factors satisfying
\eqref{eq:capacity} an admissible \emph{profile}. Nonzero scalar
multiples are identified because they vanish on the same subspaces.
Both applications use $r=L(A)$: the table already rules out fewer than
$r$ terms, so a hypothetical $r$-term decomposition has no zero term.
We need only profiles of nonzero first factors. Over $\F_2$, each
projective point has a unique nonzero representative.

\subsection{Two small-subspace facts force a common block}

Take $A=\F_2^{3\times3}$ and $r=20$. The table used below has $L(A)=20$.
The decisive improvement concerns two-dimensional subspaces
$E\subseteq A^*$ with no rank-one matrix. Without rank-one-span searches,
many subspaces had $L(E^\perp)=18$, which does not imply the
common-block reduction below. Rank-one-span searches strengthen bounds
on small restriction spaces; propagation through the existing reduction
and backtracking machinery raises all these values to $19$, enabling
the reduction. The crucial entries of the strengthened table are as
follows. Their full distributions are given in
Appendix~\ref{sec:data}.

\begin{proposition}[entries of the binary table]\label{prop:binary-table}
  For every nonzero $U\in A^*$, $L(\ker U)=19$. For every two-dimensional
  $E\subseteq A^*$,
  \[
    L(E^\perp)=
    \begin{cases}
      19,&E\text{ has no rank-one matrix},\\
      18,&E\text{ has at least one rank-one matrix}.
    \end{cases}
  \]
  The two cases contain $32{,}018$ and $11{,}417$ subspaces, respectively.
\end{proposition}

Thus every one-dimensional $E$ has capacity one, and every
two-dimensional $E$ has capacity at most two. Any admissible profile is
therefore a set of $20$ distinct nonzero matrices, with no three linearly
dependent. More importantly, if two of its matrices $U,U'$ have rank at
least two, their sum must have rank one. Otherwise
$E=\{0,U,U',U+U'\}$ would have capacity one but contain two selected
matrices.

There are $49$ rank-one matrices $xc^{\T}$, with $x,c\in\F_2^3\setminus0$.
The following $14$ three-dimensional subspaces will be called rank-one
blocks:
\[
  H_x=\{xc^{\T}:c\in\F_2^3\},\qquad
  H^c=\{xc^{\T}:x\in\F_2^3\}.
\]
Every nonzero element of a block has rank one.

\begin{lemma}[rank-one sets]\label{lem:rank-one}
  A set of rank-one matrices over $\F_2$ whose pairwise sums have rank one
  is contained in a rank-one block.
\end{lemma}
\begin{proof}
  Two distinct products $xc^{\T}$ and $yd^{\T}$ can have a rank-one sum only
  if $x=y$ or $c=d$: otherwise both pairs of vectors are independent and
  the sum has rank two. If all left vectors agree, the conclusion follows.
  Otherwise choose two elements with different left vectors and therefore
  the same right vector $c$. Any further element must share a factor with
  each of them. It cannot share both different left vectors, so its right
  vector is $c$.
\end{proof}

\begin{lemma}[common-block reduction]\label{lem:block}
  Let $W$ be the matrices of rank at least two in an admissible $20$-profile.
  Then $W\subseteq U_0+H$ for some matrix $U_0$ and rank-one block $H$.
  In particular $|W|\le8$, and the remaining $20-|W|\ge12$ matrices
  are distinct rank-one matrices.
\end{lemma}
\begin{proof}
  The assertion is immediate if $|W|\le1$. Otherwise fix $U_0\in W$.
  Each $U+U_0$, for $U\in W\setminus\{U_0\}$, has rank one by
  Proposition~\ref{prop:binary-table}. Any two of these differences have a
  rank-one sum, since their sum is the difference of two elements of $W$.
  Lemma~\ref{lem:rank-one} places all the differences in one block $H$.
\end{proof}

\subsection{Thirty-five cases and no completion}

The first-factor action of the matrix multiplication symmetries is
\[
  G=\{U\mapsto PUQ,\ U\mapsto PU^{\T}Q:
  P,Q\in\GL_3(\F_2)\},\qquad |G|=56{,}448.
\]
Sandwiching extends to a symmetry of the multiplication tensor. In the
coefficient conventions above, transposition extends as
$(U,V,W)\mapsto(U^{\T},W,V)$, interchanging the second and third factors.
Thus $G$ preserves the restriction ranks and the table, and is transitive
on the $14$ blocks.

We may fix $H=H_{e_1}$, the matrices supported on the first row. Its
stabilizer has order $4{,}032$. The other two rows index the $63$ nonzero
cosets of $H$. Enumerate the subsets of rank-at-least-two matrices in
these cosets, and identify them under $\Stab_G(H)$. Include the empty
set once. The resulting $35$ possibilities for $W$ have the distribution
\begin{center}
  \begin{tabular}{c*{9}{r}}
    \toprule
    $|W|$ & 0&1&2&3&4&5&6&7&8\\
    number of cases &1&3&5&6&8&5&4&2&1\\
    \bottomrule
  \end{tabular}
\end{center}
For each case, the task is to choose $20-|W|$ of the $49$ rank-one
matrices so that every capacity holds. The residual symmetry is the
full setwise stabilizer $\Stab_G(W)$.

Sixteen cases fail before this search starts. In each, four elements of
$W$ form an affine plane $U_0+D$, where $D$ is a two-dimensional
rank-one subspace. Their span $E$ has dimension three, and the table
gives $L(E^\perp)=17$, hence capacity three. All $9{,}114$ subspaces of
this form have that table value. The other nineteen cases require an
enumeration of their rank-one completions.

\begin{proposition}[binary enumeration]\label{prop:binary-enum}
  None of the $35$ cases has an admissible completion to a $20$-profile.
  The nineteen nontrivial searches exhaust $50{,}830$ jobs and
  $15{,}658{,}474$ nodes and return no profile.
\end{proposition}

This is a computational proposition. The algorithm is described in
Section~\ref{sec:computations}; its completeness is proved in
Appendix~\ref{sec:profile-completeness}. Appendix~\ref{sec:data} lists the
individual cases with their recorded job and node counts. The run uses
only the capacities and the common-block reduction, with no additional
inequalities coupling different factor lists. In particular it includes
$W=\varnothing$.

\begin{proof}[Proof of Theorem~\ref{thm:binary}]
  The table excludes rank below $20$. A $20$-term decomposition would give
  an admissible profile by Lemma~\ref{lem:capacity}. By
  Lemma~\ref{lem:block}, a symmetry takes its higher-rank part to one of
  the $35$ cases. Proposition~\ref{prop:binary-enum} excludes every
  rank-one completion of every case. Hence no such decomposition exists.
\end{proof}

The enumeration uses the entire strengthened table, not just the entries stated in
Proposition~\ref{prop:binary-table}.

\section{The two computations supporting the proof}\label{sec:computations}

\subsection{Producing the restriction bounds}

We use the framework of Wang~\cite{wang2026automated} to enumerate
restriction spaces $S\subseteq A$ up to symmetry and assign a lower
bound $L(S)$ to each orbit. It processes spaces in increasing dimension,
so bounds on smaller restrictions are available for subsequent steps.
The four inherited techniques are flattening, degenerate reduction,
forced products, and substitution with backtracking. Each successful
step records its technique and the data needed to check it. The verifier
recomputes elementary bounds, checks reduction witnesses, and replays
backtracking traces. We do not repeat these algorithms here.

The additional search tests whether the slice space can lie in a
rank-one-spanned space of a prescribed dimension. It strengthens
individual restrictions; the existing techniques then propagate these
improvements. The final binary table contains $496$ orbits, expanding to
$8{,}283{,}458$ subspaces. The ternary table contains $31$ orbits and
$56{,}632$ subspaces. The latter uses the same finite-field framework
with the $\F_3$ matrix multiplication tensor and its symmetry action.

\subsection{Searching for a rank-one-spanned extension}\label{sec:slice-search}

For this subsection, write a tensor in any chosen orientation as
$T\in D\otimes B'\otimes C'$. Let $V\subseteq B'\otimes C'$ be the
image of contraction by $D^*$, and let $\rho=\dim V$. Let $B_T,C_T$
be the minimal factor spaces supporting $V$, which we call the
\emph{core factor spaces}, and put
$K=B_T\otimes C_T$. The following classical characterization is the
basis of the search~\cite[Prop.~14.45]{burgisser1997algebraic}.

\begin{proposition}[rank via the slice space]\label{prop:slices}
  The rank of $T$ is the least dimension of a subspace $U$ spanned by
  rank-one matrices and satisfying $V\subseteq U\subseteq K$.
\end{proposition}
\begin{proof}
  A tensor decomposition places all slices in the span of its second-third
  products. Conversely, a rank-one basis of a space containing $V$ expresses
  every slice and gives a decomposition with that many terms. It suffices
  to work inside $K$: projections onto $B_T$ and $C_T$ fix $V$, preserve
  matrix rank at most one, and cannot increase the dimension of the span.
\end{proof}

Set $Q=K/V$ and let $\pi:K\to Q$. A space containing $V$ is uniquely
$\pi^{-1}(W)$ for a subspace $W\subseteq Q$, with
$\dim\pi^{-1}(W)=\rho+\dim W$. Thus, for a target rank $r\ge\rho$, we
need only test dimensions $0\le k\le e:=r-\rho$. For a given $W$ the
test is exact:
\begin{equation}\label{eq:span-test}
  \dim\spn\{M\in K:\rank M=1,\ \pi(M)\in W\}=\rho+k.
\end{equation}
If no $W$ of these dimensions passes, $\R(T)\ge r+1$.

The exhaustive implementation first groups the rank-one matrices by
their quotient keys $\pi(M)$, using projective keys over odd primes.
It enumerates every $W$ in reduced row echelon form. A count of the
available matrices rejects many candidates before Gaussian elimination
tests \eqref{eq:span-test}. The span of the key-zero matrices is inserted
once at the start. If $q=\dim Q$, the number of candidates is
$\sum_{k=0}^e \genfrac{[}{]}{0pt}{}{q}{k}_{|\F|}$, a sum of Gaussian binomial
coefficients. Over $\F_2$ this is $831{,}982$ for $(q,e)=(9,3)$, but
about $2.3\cdot10^{17}$ for $(q,e)=(18,4)$.

For the larger binary quotients a second implementation exploits the
product structure of rank-one matrices. Fixing a nonzero vector in the
smaller core factor ($B_T$ or $C_T$) gives a \emph{family}; its quotient keys form a
linear subspace of $Q$. Any successful $W$ must meet these family images
in enough dimensions to span $\pi^{-1}(W)$ by rank-one matrices. For
$e\le5$, under explicit numerical conditions, these intersections give
a smaller covering list of candidates. Appendix~\ref{sec:family} proves
coverage and the safety of each filter. Neither implementation treats
an exhausted budget or an unsupported case as an exclusion.

\subsection{Enumerating profiles}\label{sec:profile-algorithm}

The second computation searches directly over projective first factors.
The input is the expanded table $L$, a target $r$, a candidate set, and
a group preserving the table and candidate set. Both tables are
monotone: $S'\subseteq S$ implies $L(S')\le L(S)$.
In both applications every candidate has point capacity at most one,
so the profiles are sets. In the binary case the higher-rank part $W$
is fixed first; in the ternary case the point capacities themselves
restrict all candidates to rank-one matrices.

The search adds candidates in a fixed increasing order. For the
selected set it maintains the distinct spans of subsets, their reduced
row echelon bases, and their occupancy counts. A violated capacity
rejects the branch. A saturated span prevents the addition of further
points inside it, and the search also stops if too few available points
remain. Lexicographic canonicity under the relevant group removes
symmetry-equivalent branches.
\begin{quote}
  \begin{minipage}{\linewidth}\small
    Search(selected, future):\\
    \hspace*{1em}reject a violated capacity or a noncanonical selected set;\\
    \hspace*{1em}if $r$ points are selected, check all remaining capacities and report;\\
    \hspace*{1em}discard future points in saturated spans;\\
    \hspace*{1em}reject if too few future points remain;\\
    \hspace*{1em}for each next point in increasing order, add it and recurse.
  \end{minipage}
\end{quote}
This is a description of the search logic, not a substitute for its
completeness proof. Appendix~\ref{sec:profile-completeness} explains why
spans of selected points detect every violation and why canonical
prefixes suffice. Splitting the search at a fixed depth produces
independent jobs without changing the search space.

An enumeration establishes nonexistence or a complete classification
only if every job finishes without a node-budget or solution-limit
stop. A reported feasible profile is not evidence that its job was
exhausted. The runs used here finish well below both limits; the
numerical checks are recorded in Appendix~\ref{sec:data}.

\section{\texorpdfstring{$\R_{\F_3}(\mmt{2}{3}{3})\ge15$}{Tensor rank over F3 of <2,3,3> is at least 15}}\label{sec:ternary}

We now take $A=\F_3^{2\times3}$, $B=\F_3^{3\times3}$,
$C=\F_3^{2\times3}$, and $r=14$. The table has $L(A)=14$ and
\begin{equation}\label{eq:ternary-points}
  L(\ker U)=
  \begin{cases}13,&\rank U=1,\\14,&\rank U=2.
  \end{cases}
\end{equation}
Consequently every admissible $14$-profile consists of distinct
rank-one points. They form
$\PG(1,3)\times\PG(2,3)$, a set of $4\cdot13=52$ points.
The full table is given in Appendix~\ref{sec:data}.

\subsection{The three surviving profiles}

The generic enumerator, using sandwich symmetries from
$\GL_2(\F_3)\times\GL_3(\F_3)$, leaves exactly three orbits of
profiles. Its projective permutation group has order $134{,}784$.
The computation takes $3$ jobs and $7{,}024$ nodes.

Here are representatives that make the exclusion transparent. Write
\[
  e_1=(1,0)^{\T},\quad e_2=(0,1)^{\T},\quad
  f=(1,1)^{\T},\quad g=(1,2)^{\T},
\]
and take the column directions
$c_1=(1,0,0)^{\T}$, $c_2=(0,1,0)^{\T}$,
$c_3=(0,0,1)^{\T}$, $c_4=(1,1,1)^{\T}$.
In the array below, a bullet denotes a point common to all three
profiles, and $P_i$ denotes the single extra point of profile $P_i$.
No points with other column directions occur.
\begin{center}
  \begin{tabular}{c c c c c}
    \toprule
    $xc^{\T}$ & $c_1$ & $c_2$ & $c_3$ & $c_4$\\
    \midrule
    $e_1$ & $\bullet$ & $\bullet$ & $\bullet$ & $\bullet$\\
    $e_2$ & $\bullet$ & $\bullet$ & $\bullet$ & $\bullet$\\
    $f$   & $\bullet$ & $\bullet$ & $\bullet$ & $P_1$\\
    $g$   & $\bullet$ & $\bullet$ & $P_2$ & $P_3$\\
    \bottomrule
  \end{tabular}
\end{center}
Each profile therefore has thirteen common points and one extra point.
The completeness of this computational classification uses the same
argument as the binary enumeration; the unrestricted symmetry check is
described in Section~\ref{sec:verification}, and the recorded counts are
given in Appendix~\ref{sec:data}.

\subsection{Tight and one-excess restrictions}

For $z\in\F_3^3\setminus0$, restrict the first input to
\[
  S_z=\{vz^{\T}:v\in\F_3^2\}.
\]
The restricted multiplication is $(v,Y)\mapsto v(z^{\T}Y)$, an
outer-product tensor with output support all of $C$, of dimension six.
Its second-factor support is
\[
  H_z=\{zq^{\T}:q\in\F_3^3\}\subseteq B^*.
\]
A term with first factor $x_sc_s^{\T}$ is active exactly when
$c_s^{\T}z\ne0$, and its restricted first factor is
$a_s=(c_s^{\T}z)x_s$. If $z,z'$ are not proportional, then
$H_z\cap H_{z'}=0$.

The next lemma adapts D'Ambrosio's tight and one-excess arguments to
the column restrictions above. Here the field is $\F_3$, and the
second and third factors exchange roles relative to his proof.

\begin{lemma}[tight and one-excess restrictions, after D'Ambrosio~{\cite[Section~3]{dambrosio2026exact}}]\label{lem:restriction}
  Consider the active terms of a decomposition restricted to $S_z$.
  \begin{enumerate}
    \item If there are six active terms, their output factors $W_s$ form a
      basis of $C$, every $V_s$ lies in $H_z$, and the dual output basis has
      the form $\omega_s=p_sq_s^{\T}$, with $p_s$ proportional to $x_s$ and
      $V_s$ proportional to $zq_s^{\T}$.
    \item If there are seven active terms, the images of their second
      factors in $B^*/H_z$ span a space of dimension at most one.
  \end{enumerate}
\end{lemma}
\begin{proof}
  The output flattening of $T_{S_z}$ has rank six, so the active $W_s$
  span $C$. For six terms they are a basis. Contracting with its dual
  vector $\omega_s\in C^*$ isolates $a_s\otimes V_s$. If $M_s$ is the
  $2\times3$ coefficient matrix of $\omega_s$, contracting the
  outer-product tensor gives the bilinear form
  $v^{\T}M_s(z^{\T}Y)^{\T}$. Its expression as the nonzero simple
  tensor $a_s\otimes V_s$ forces $M_s=p_sq_s^{\T}$ with $p_s$ proportional
  to $a_s$, and $V_s$ proportional to $zq_s^{\T}$.

  For seven terms there is precisely one independent relation among the
  $W_s$, say $\sum_s\kappa_sW_s=0$. Let $\pi_z:B^*\to B^*/H_z$.
  Projecting the restricted tensor gives
  $\sum_s a_s\otimes\pi_z(V_s)\otimes W_s=0$. Since the relation space
  of the $W_s$ is one-dimensional, there is a tensor $Z$ such that
  \[
    a_s\otimes\pi_z(V_s)=\kappa_s Z\qquad\text{for every active }s.
  \]
  If $Z=0$, every image vanishes. Otherwise a nonzero $\kappa_s$ shows
  that $Z$ is simple, and all nonzero $\pi_z(V_s)$ are proportional to
  its second factor. This proves the second assertion.
\end{proof}

The two overlap arguments below adapt D'Ambrosio's profile
exclusions~\cite[Section~5]{dambrosio2026exact} to these column
restrictions.

\begin{corollary}[overlapping restrictions]\label{cor:overlap}
  Let $z,z'\in\F_3^3$ be nonproportional.
  \begin{enumerate}
    \item Two six-term restrictions, to $S_z$ and to $S_{z'}$, share no
      term.
    \item If a six-term restriction to $S_z$ and a seven-term restriction
      to $S_{z'}$ share terms, the second factors of the shared terms are
      proportional.
  \end{enumerate}
\end{corollary}
\begin{proof}
  A term shared by two six-term restrictions has its nonzero $V_s$ in
  $H_z\cap H_{z'}=0$ by the first part of Lemma~\ref{lem:restriction}.
  In the second case the shared $V_s$ lie in $H_z$, where the quotient map
  to $B^*/H_{z'}$ is injective, and by the second part of
  Lemma~\ref{lem:restriction} their quotient images span at most one
  dimension.
\end{proof}

\subsection{Excluding the profiles}

For $P_2$, the restrictions with $z=(1,0,0)^{\T}$,
$(0,1,0)^{\T}$, and $(0,0,1)^{\T}$ each have six active terms.
The first two share the terms $(e_1,c_4)$ and $(e_2,c_4)$, contradicting
the first part of Corollary~\ref{cor:overlap}.

For both $P_1$ and $P_3$, take
\[
  z=(0,0,1)^{\T},\qquad z'=(0,1,2)^{\T}.
\]
The first restriction has six active terms, those in columns $c_3,c_4$.
The second has seven, those in columns $c_2,c_3$, because
$c_1^{\T}z'=c_4^{\T}z'=0$. Their common terms are
$(e_1,c_3),(e_2,c_3),(f,c_3)$, whose second factors are proportional by
the second part of Corollary~\ref{cor:overlap}.
In the six-term restriction their dual output basis vectors are
therefore proportional, respectively, to
\[
  e_1q^{\T},\qquad e_2q^{\T},\qquad (e_1+e_2)q^{\T}
\]
for a common nonzero $q$. These three vectors are linearly dependent,
contradicting that they belong to a basis.

\begin{proof}[Proof of Theorem~\ref{thm:ternary}]
  The table excludes rank below $14$. Every $14$-term decomposition
  would give one of $P_1,P_2,P_3$ up to symmetry, and all three have just
  been excluded. Thus the rank is at least $15$. Hopcroft and Kerr's
  algorithm~\cite{hopcroft1971minimizing} gives the matching upper bound.
\end{proof}

\section{Verification and limitations}\label{sec:verification}

These are computer-assisted proofs with two different verification
tasks. The subspace lower-bound tables are accompanied by proof records
checked by the verifier in Wang's framework~\cite{wang2026automated}.
The profile classifications are obtained by exhaustive searches, whose
completeness is proved in Appendix~\ref{sec:profile-completeness} and
supported by computational cross-checks. These searches do not yet have
traces replayed by a separate small checker.

For the first task, the trusted code consists of the field and linear
algebra routines, the tensor and symmetry action, the proof checkers,
and the verifier's sweep through the restriction spaces. The programs
that discover bounds and choose their proofs are outside this boundary.
There is an important qualification for rank-one-span records: the
verifier reruns the same exclusion engine used by the search, rather
than checking an independently generated proof of nonexistence. Its
implementation and the coverage proof in Appendix~\ref{sec:family} are
therefore part of the trust boundary. Expanding the orbit table to all
subspaces is another computation used by the profile search.

The binary verification takes about $8$ hours on $32$ cores; the
ternary verification takes seconds.
The companion repository identifies the artifacts and gives the commands
that reproduce every computation, with their expected outputs.

The following checks support the computations.
\begin{itemize}
  \item Monotonicity was checked on every covering pair of restriction
    spaces: $213{,}188{,}689$ pairs for the binary table and $969{,}696$
    for the ternary table, with no violation. Orbit counts, point-rank
    classes, the dichotomy in Proposition~\ref{prop:binary-table}, and the affine-plane rejections
    were also checked against the expanded tables.
  \item Removing inner lexicographic symmetry pruning gives the same
    negative answer in all $35$ binary cases, after $2{,}155{,}604{,}077$
    nodes. This check retains the structural reduction and the $35$
    symmetry-reduced outer cases. Over $\F_3$, removing symmetry pruning
    on the full rank-one candidate set takes $52{,}972{,}264$ nodes and
    gives the same three profile orbits.
  \item Positive controls include all three factor lists of a known
    $23$-term binary decomposition, which satisfy all capacities, and a
    restricted binary instance with $232$ profiles at $r=22$, for which
    searches with and without symmetry pruning agree. The generic
    enumerator also reproduces D'Ambrosio's classification~\cite[Section~4]{dambrosio2026exact}
    of $252$ supports in four orbits
    of sizes $63,21,126,42$ for the transposed $\mmt{3}{2}{4}$ instance.
  \item The rank-one-span implementations were compared on small
    instances where both complete, and tested on tensors with known
    decompositions and planted rank-one-spanned extensions. Such tests
    check implementation behavior; they do not replace the coverage proof.
  \item Over $\F_2$, every subspace lower bound is at most the corresponding
    upper bound obtained by flip-graph search~\cite{kauers2023flip}.
\end{itemize}

The bounds remain limited by the strength of the subspace lower-bound tables
and by the cost of both searches. A surviving profile need not lift to
a tensor decomposition, as the ternary case illustrates. Similarly,
failure to finish a rank-one-span search says nothing about whether its
current lower bound is tight. At the next binary target, $r=21$, every
capacity increases by one, and the argument forcing pairwise
rank-one differences no longer applies. This paper makes no claim
about existence or nonexistence of $21$-profiles. Useful next steps are
stronger restriction bounds and independently checkable traces for
the completed profile searches.


\begingroup\raggedright
\bibliography{refs}
\endgroup

\appendix

\section{Completeness and pruning of profile enumeration}\label{sec:profile-completeness}

We prove completeness for the searches used here. The table is invariant
under the search group and monotone in $S$. Points have a fixed total
order, and all capacities count selected points with multiplicity when
the implementation is used on multisets.

\begin{lemma}[capacity pruning]\label{lem:pruning}
  Rejecting a partial selection with a violated capacity cannot discard an
  admissible completion. Moreover, if a completed selection violates a
  capacity, it violates one on a subspace spanned by selected points.
\end{lemma}
\begin{proof}
  Occupancy cannot decrease as points are added, while capacities are
  fixed. For the second assertion, suppose $E$ is violated and let $E'$
  be the span of the selected points in $E$. Then $E'\subseteq E$ and
  the two occupancies are equal. Monotonicity gives
  $L((E')^\perp)\ge L(E^\perp)$, hence $c_r(E')\le c_r(E)$.
  Thus $E'$ is also violated.
\end{proof}

The maintained span collection is sufficient by induction. After adding
a point $u$, a span of selected points is either an old span or
$E+\langle u\rangle$ for an old span $E$, including $E=0$. Old spans
containing $u$ have their counts incremented; a new span is counted
against the entire selection, not just its displayed generators.
Equal reduced row echelon bases are merged. Undo records restore this
state on backtracking. The generic search maintains dimensions
$1,\ldots,\dim A-1$. The binary search maintains dimensions through
six while branching and checks dimensions seven and eight at a leaf.
This delayed checking can add work but cannot cause a false rejection.
The full dual space has capacity $r$, since $L(0)=0$, and never rejects
a selection of at most $r$ nonzero points.

A saturated span cannot contain any further selected point. Marking its
remaining points unavailable is therefore safe. Summing the remaining
point multiplicities over all available future points gives an upper
bound on the number of possible additions. This may overestimate what
can be added simultaneously, which is harmless: a branch is rejected
only when even this upper bound is too small.

\begin{lemma}[canonical prefixes]\label{lem:prefix}
  For a permutation group acting on an ordered point set, every prefix
  of a lexicographically least sorted representative of a set or multiset
  orbit is lexicographically least in its own orbit.
\end{lemma}
\begin{proof}
  Let $P$ be the first $k$ entries of a sorted list $F$. If some group
  element makes the sorted list $gP$ smaller than $P$, let $i$ be their
  first differing position. The $j$th entry of sorted $gF$ is at most
  the $j$th entry of sorted $gP$ for every $j\le k$, since $gP$ is a
  submultiset of $gF$. Hence $gF$ is smaller than $F$, either before
  position $i$ or at that position. This contradicts the minimality of
  $F$.
\end{proof}

It follows that rejecting noncanonical prefixes retains a path to the
least representative of every admissible orbit. With a fixed binary
higher-rank part $W$, apply the lemma to the growing rank-one part
under $\Stab_G(W)$, which preserves both the rank-one candidates and
the capacities with $W$ already inserted. The outer reduction is
complete because every profile has a block by Lemma~\ref{lem:block},
the blocks form one $G$-orbit, and every subset of the higher-rank
elements in each fixed-block coset is considered. A profile may admit
more than one block, so the outer reduction can create redundancy;
it cannot omit a profile. Final canonicalization under $G$ removes
duplicate output orbits where necessary.

Finally, splitting at depth $d$ enumerates all retained prefixes of
length $d$ and assigns their continuations to separate jobs; shorter
terminal branches are handled during splitting. The increasing point
order assigns each remaining branch to a unique prefix. Therefore,
subject to completing every job without an early stop, the combined
search reports precisely the admissible profiles up to the chosen
symmetry. This proves the completeness used in
Proposition~\ref{prop:binary-enum} and in the ternary classification.

\section{Family-based search for rank-one-spanned extensions}\label{sec:family}

This section justifies the larger-quotient binary search introduced in
Section~\ref{sec:slice-search}. All dimensions below are vector-space
dimensions; a \emph{point} has dimension one and a \emph{line} has
dimension two. Work over $\F_2$ and retain $V\subseteq K$, $Q=K/V$,
$\pi:K\to Q$, and $\rho=\dim V$.

\subsection{Necessary conditions and safe completion}

Write $K=F_0\otimes O$, choosing $\dim F_0$ to be the smaller core
dimension. Each nonzero $f\in F_0$ defines a family and a linear map
\[
  \varphi_f:O\to Q,\qquad o\mapsto\pi(f\otimes o),\qquad
  I_f=\im\varphi_f.
\]
There are $N=2^{\dim F_0}-1$ families. Call a nonzero key \emph{occupied}
if it belongs to some $I_f$, and put $c(q)=\#\{f:q\in I_f\}$. Let
$b_0$ be the span of all rank-one matrices in $V$, and let $d_0=\dim b_0$.
A \emph{witness} of dimension $k$ is a subspace $W\subseteq Q$ for
which $\pi^{-1}(W)$ is rank-one spanned. Set
\[
  d_f(W)=\dim(I_f\cap W),\qquad n=\rho+k-d_0.
\]

\begin{lemma}\label{lem:family-count}
  Every witness $W$ is spanned by its occupied keys, and
  \begin{equation}\label{eq:family-count}
    \sum_f d_f(W)\ge n,\qquad
    \sum_{q\in W\setminus0}c(q)
    =\sum_f(2^{d_f(W)}-1)\ge n.
  \end{equation}
\end{lemma}
\begin{proof}
  Apply $\pi$ to a rank-one spanning set of $\pi^{-1}(W)$ to obtain
  the first assertion. Within family $f$, the allowable vectors form
  $\varphi_f^{-1}(W)$, of dimension
  $\dim\ker\varphi_f+d_f(W)$. The kernel contribution lies in $b_0$,
  so this family contributes at most $d_f(W)$ dimensions modulo $b_0$.
  Summing gives $\rho+k\le d_0+\sum_f d_f(W)$. Counting incidences of
  families and nonzero keys gives the displayed equality, and
  $2^d-1\ge d$ gives the remaining inequality.
\end{proof}

We call an initial subspace $X\subseteq Q$ that the search extends to a
candidate $W$ a \emph{root}. If an enumerated root $X$ lies in a witness
$W$, occupied keys in $W$
can be chosen to complete a basis modulo $X$. Thus every root contained
in a witness has a completion using only occupied keys.
There is also a useful filter for the last generator. If
$W=X\oplus\langle g\rangle$ and the current case assumes $d_f(W)\le h$,
then
\begin{equation}\label{eq:family-growth}
  d_f(W)\le
  \min\{d_f(X)+[g\in X+I_f],\ h,\ \dim I_f\}.
\end{equation}
Here $[g\in X+I_f]$ is $1$ if $g\in X+I_f$ and $0$ otherwise.
Indeed, adjoining one vector increases an intersection dimension by
at most one, and an increase requires $g\in X+I_f$. Rejecting when
the sum of these upper bounds is less than $n$ is safe by
\eqref{eq:family-count}. The cap $h$ is used only in cases for which
it is justified.

For a completed candidate, the program tests the incidence count in
\eqref{eq:family-count}, then the exact sum of intersection dimensions,
and finally \eqref{eq:span-test} by Gaussian elimination, starting with
$b_0$. The first two tests are only filters. Only the last decides
whether a candidate is a witness.

\subsection{The covering cases through dimension five}

The search first tests $k=0$ by checking whether $d_0=\rho$, then tests
successive dimensions up to $e=r-\rho$. Duplicate candidates do not
affect correctness. Whenever
a root below is completed, all independent occupied-key completions
are considered, subject only to the safe filters above.

For $k=1$, enumerate occupied keys with $c(q)\ge n$; for $k=2$,
enumerate pairs of independent occupied keys. Lemma~\ref{lem:family-count}
proves coverage. For $k\ge3$ let $d_{\max}=\max_f d_f(W)$. Two cases
apply at every dimension:
\begin{enumerate}
  \item If $d_{\max}=k$, enumerate all $k$-subspaces of every $I_f$.
  \item If $d_{\max}=k-1$, enumerate every $(k-1)$-subspace of every
    $I_f$ and complete it by one occupied key. The growth filter uses
    $h=k-1$.
\end{enumerate}
The witness itself supplies the indicated root in each case. For
$k=3$, the remainder has $d_f\le1$, hence $\sum_f d_f\le N$.
This is impossible when $n>N$, which is the search's required condition
for this remainder. If the condition fails, the family engine cannot
claim an exclusion.

\paragraph{Dimension four.}
The remainder has $d_f\le2$. Require $n\ge N+3$. There are at least
$n-N\ge3$ families for which $D_f=I_f\cap W$ is a line, since
$\sum_f d_f\le N+\#\{f:d_f=2\}$. Count these lines with family
multiplicity. The following cases cover every arrangement.
\begin{enumerate}
  \item[$\gamma$.] A line occurs in at least three family images.
    Enumerate that line and add two occupied keys.
  \item[$\alpha$.] Otherwise, two distinct family lines meet in a point.
    Enumerate pairs of family lines through a point of multiplicity at
    least two. Their span has dimension three; add one occupied key.
  \item[$\beta$.] Otherwise, the distinct lines are pairwise disjoint.
    If there are at least three, choose $D_1,D_2,D_3$. Then
    $W=D_1\oplus D_2$. A nonzero $p\in D_3$ has a unique expression
    $p=x+y$, with nonzero $x\in D_1$, $y\in D_2$. Enumerate the occupied
    triples $x,y,x+y$, a line in one family through $x$, and a line in
    another family through $y$. The pair $D_1,D_2$ is included and spans $W$.
  \item[$\delta$.] The only remaining possibility is two distinct
    disjoint lines, one occurring exactly twice. Enumerate lines in
    exactly two family images and a disjoint line in another family.
    Their direct sum is $W$.
\end{enumerate}
These cases are exhaustive: if neither $\gamma$ nor $\alpha$ applies,
each distinct line occurs at most twice and the lines are pairwise
disjoint. At least three distinct lines give $\beta$; exactly two
give $\delta$ because at least three occur with multiplicity. A single
distinct line would have given $\gamma$. The last-generator growth
filters in this remainder use $h=2$.

\paragraph{Dimension five.}
The remainder has $d_f\le3$. Require $n\ge N+4$, and write
$t_i=\#\{f:d_f=i\}$. Equation~\eqref{eq:family-count} implies
\[
  2t_3+t_2\ge n-N\ge4.
\]
We use the program's case labels to identify the following roots.

If $t_3\ge2$, choose two three-dimensional intersections from distinct
families. They meet nontrivially inside the five-dimensional $W$.
Case C3 enumerates pairs of three-subspaces of family images through
a common occupied point of multiplicity at least two. Distinct such
spaces span dimension four or five, so add one or zero keys,
respectively. Coincident spaces are covered by C3$'$: enumerate
three-subspaces of pairwise family-image intersections and add two keys.

If $t_3=1$, let $D$ be that three-dimensional intersection. There are
at least $n-N-2$ other families whose intersections are lines.
If one such line meets $D$ in exactly a point, C4 enumerates the
three-space and the line through their common point; their span has
dimension four, so add one key. If a line is disjoint from $D$,
it is disjoint from any line inside $D$. Case C6 enumerates disjoint
lines from distinct families and adds one key to their four-dimensional
span. If all the lines lie inside $D$, case C9 enumerates
three-subspaces of an image containing lines from at least $n-N-2$
other families and adds two keys. Thus all possibilities with $t_3=1$
are covered.

If $t_3=0$, there are at least $n-N\ge4$ family lines. A disjoint
pair gives C6. Otherwise the distinct lines pairwise intersect.
If only one distinct line occurs, C10 enumerates a line lying in at
least $n-N$ images and adds three keys. If at least two distinct lines
occur, choose two and let their intersection be $p$ and their span be $X$, of dimension
three. Any further line meeting both either lies in $X$ or passes
through $p$. A line outside $X$ meets $X$ only at $p$, so its presence
forces every family line inside $X$ to pass through $p$ as well.
Consequently all lines either lie in a common three-space or pass
through a common point.

In the common-three-space case, C7 enumerates the spans of two
intersecting family lines that contain lines from at least $n-N$
families, and adds two keys. In the common-point case, consider the
line directions modulo that point. If they span at least three
dimensions, C8 enumerates the common point and three independent
line directions from distinct families, then adds one key to this
four-dimensional root. The point lies in at least $n-N$ images.
If the directions span two dimensions, C7 applies; if they span one,
C10 applies. This proves coverage in the final remainder. Growth
filters here use $h=3$. The implementation's additional case C5 is a
fast path covered already by C4, C6, and C9; it is not needed for
the coverage argument.

\begin{proposition}[completed family search]\label{prop:family-complete}
  Suppose $0\le e=r-\rho\le5$, the numerical conditions above hold
  for every reached remainder, and the family search finishes without
  exhausting its operation budget. It finds a witness if and only if
  $\R(T)\le r$. A completed negative search therefore proves
  $\R(T)\ge r+1$.
\end{proposition}
\begin{proof}
  A found witness gives a decomposition by
  Proposition~\ref{prop:slices}. Conversely, if the rank is at most $r$,
  choose a witness of smallest dimension $k\le e$. The initial test
  finds it if $k=0$. Otherwise the search reaches $k$, and one of the
  covering cases contains a root lying in the witness. Occupied keys
  complete that root to the witness; Lemma~\ref{lem:family-count} and
  \eqref{eq:family-growth} show that every filter retains this completion.
  The final exact test accepts it.
\end{proof}

The hard binary restrictions treated by this engine, after the chosen
cyclic orientation, have $\rho=9$, core dimensions $3$ and $9$,
$\dim Q=18$, $N=7$, and $d_0\in\{0,3\}$. At $k=3,4,5$ the smallest
values of $n$ are $9,10,11$, respectively, so all the numerical
conditions hold. These parameters describe the relevant family-search
instances, not every restriction in the table. Unsupported parameters
or an exceeded budget give no claim; an affordable exhaustive search
may be used instead. Verification recomputes an exclusion at the
recorded target $L(S)-1$, in the recorded orientation, with an
operation budget of $4\cdot10^{13}$.

\section{Tables and outer cases}\label{sec:data}

\subsection{The subspace lower-bound tables}

Tables~\ref{tab:binary} and~\ref{tab:ternary} give the distributions
of the bounds over all restriction spaces $S$. An entry $b\ (m)$
means that exactly $m$ spaces of that dimension have $L(S)=b$.
The bounds are constant on the symmetry orbits; these are distributions
of lower bounds, not claims of exact rank on every orbit.

\begin{table}[!htbp]
  \centering\small
  \begin{tabular}{@{}rr>{\raggedright\arraybackslash}p{11.6cm}@{}}
    \toprule
    $\dim S$ & spaces & $L(S)$ (count)\\
    \midrule
    0 & 1 & 0 (1)\\
    1 & 511 & 3 (49), 6 (294), 9 (168)\\
    2 & 43,435 & 6 (980), 9 (15,827), 12 (25,284), 14 (1,344)\\
    3 & 788,035 & 9 (10,157), 12 (276,458), 13 (672), 14 (271,362),
    15 (226,842), 16 (2,352), 17~(192)\\
    4 & 3,309,747 & 11 (49), 12 (56,938), 13 (3,430), 14 (368,088),
    15 (1,590,981), 16 (1,219,365), 17~(70,896)\\
    5 & 3,309,747 & 14 (8,722), 15 (152,341), 16 (1,108,576), 17 (2,040,108)\\
    6 & 788,035 & 15 (14), 16 (7,497), 17 (200,312), 18 (580,212)\\
    7 & 43,435 & 18 (11,417), 19 (32,018)\\
    8 & 511 & 19 (511)\\
    9 & 1 & 20 (1)\\
    \bottomrule
  \end{tabular}
  \caption{The binary $\mmt{3}{3}{3}$ table: $8{,}283{,}458$ subspaces
  in $496$ orbits.}\label{tab:binary}
\end{table}

\begin{table}[!htbp]
  \centering\small
  \begin{tabular}{@{}rrl@{}}
    \toprule
    $\dim S$ & spaces & $L(S)$ (count)\\
    \midrule
    0 & 1 & 0 (1)\\
    1 & 364 & 3 (52), 6 (312)\\
    2 & 11,011 & 6 (1,001), 9 (10,010)\\
    3 & 33,880 & 9 (7,544), 12 (26,336)\\
    4 & 11,011 & 10 (13), 12 (10,348), 13 (650)\\
    5 & 364 & 13 (52), 14 (312)\\
    6 & 1 & 14 (1)\\
    \bottomrule
  \end{tabular}
  \caption{The ternary $\mmt{2}{3}{3}$ table: $56{,}632$ subspaces
  in $31$ orbits.}\label{tab:ternary}
\end{table}

In the binary table, $\BinaryFlattenCount$ entries are proved by flattening, $\BinaryDegenerateCount$ by
degenerate reduction, $\BinaryForcedCount$ by forced products, $\BinaryBacktrackingCount$ by backtracking,
and $\BinaryRankOneSpanCount$ by rank-one-span searches; the ternary counts are
$\TernaryFlattenCount,\TernaryDegenerateCount,\TernaryForcedCount,\TernaryBacktrackingCount,\TernaryRankOneSpanCount$. These counts describe the final proof of each entry, not
every technique used to obtain it: the binary improvements came from
rank-one-span searches on small restrictions followed by propagation
through degenerate reduction and backtracking, and the ternary
rank-one-span searches used the exhaustive engine and were likewise
followed by propagation.

In terms of the subspaces $E\subseteq A^*$ of Section~\ref{sec:binary},
the binary table has $L(E^\perp)=19$ for all $511$ one-dimensional $E$,
the values of Proposition~\ref{prop:binary-table} for the $43{,}435$
two-dimensional $E$, and $L(E^\perp)=17$ for the $9{,}114$
three-dimensional $E$ used for the immediate outer rejections. In the
ternary table, the $52$ rank-one projective points $U$ have
$L(\ker U)=13$ and the $312$ rank-two points have $L(\ker U)=14$, as in
\eqref{eq:ternary-points}.

\subsection{The binary outer cases and the ternary output}

Encode a binary first-factor matrix as
$\sum_{i,j=0}^2 U_{ij}2^{3i+j}$. Thus integers $1,\ldots,7$ are the
nonzero elements of the fixed first-row block. Table~\ref{tab:outer}
lists all cases requiring search. The other cases are rejected by the
affine-plane capacity before branching. For completeness their
higher-rank parts are listed in Table~\ref{tab:rejected}; the four
possible offending planes are
\[
  \begin{aligned}
    A_0&=\{10,11,12,13\},& A_1&=\{80,81,82,83\},\\
    A_2&=\{80,81,84,85\},& A_3&=\{84,85,86,87\}.
  \end{aligned}
\]
Every row of Table~\ref{tab:rejected} contains its indicated plane.

\begin{table}[!htbp]
  \centering\small
  \begin{tabular}{@{}rrlrr@{}}
    \toprule
    case & $|W|$ & $W$ & jobs & nodes\\
    \midrule
    0 & 0 & $\varnothing$ & 17 & 52,233\\
    1 & 1 & 10 & 520 & 1,055,973\\
    2 & 1 & 80 & 692 & 1,827,677\\
    3 & 1 & 84 & 292 & 972,527\\
    4 & 2 & 10 11 & 662 & 198,683\\
    5 & 2 & 10 12 & 3,592 & 1,733,655\\
    6 & 2 & 80 81 & 3,950 & 2,263,334\\
    7 & 2 & 80 84 & 4,106 & 3,171,559\\
    8 & 2 & 84 85 & 1,622 & 1,043,591\\
    9 & 3 & 10 11 12 & 2,301 & 277,935\\
    10 & 3 & 10 12 14 & 2,253 & 215,795\\
    11 & 3 & 80 81 82 & 1,621 & 118,954\\
    12 & 3 & 80 81 84 & 7,062 & 1,028,501\\
    13 & 3 & 80 84 85 & 8,385 & 1,069,256\\
    14 & 3 & 84 85 86 & 1,362 & 133,278\\
    16 & 4 & 10 11 12 14 & 1,561 & 56,025\\
    18 & 4 & 80 81 82 84 & 2,582 & 90,540\\
    20 & 4 & 80 81 84 86 & 5,520 & 244,230\\
    21 & 4 & 80 84 85 86 & 2,730 & 104,728\\
    \midrule
    \multicolumn{3}{l}{Total, all infeasible} & 50,830 & 15,658,474\\
    \bottomrule
  \end{tabular}
  \caption{The nineteen searched binary cases. Every search finishes
  without an admissible profile.}\label{tab:outer}
\end{table}

\begin{table}[!htbp]
  \centering\small
  \begin{tabular}{@{}rrlc@{}}
    \toprule
    case & $|W|$ & $W$ & offending plane\\
    \midrule
    15 & 4 & 10 11 12 13 & $A_0$\\
    17 & 4 & 80 81 82 83 & $A_1$\\
    19 & 4 & 80 81 84 85 & $A_2$\\
    22 & 4 & 84 85 86 87 & $A_3$\\
    23 & 5 & 10 11 12 13 14 & $A_0$\\
    24 & 5 & 80 81 82 83 84 & $A_1$\\
    25 & 5 & 80 81 82 84 85 & $A_2$\\
    26 & 5 & 80 81 84 85 86 & $A_2$\\
    27 & 5 & 80 84 85 86 87 & $A_3$\\
    28 & 6 & 10 11 12 13 14 15 & $A_0$\\
    29 & 6 & 80 81 82 83 84 85 & $A_1$\\
    30 & 6 & 80 81 82 84 85 86 & $A_2$\\
    31 & 6 & 80 81 84 85 86 87 & $A_2$\\
    32 & 7 & 80 81 82 83 84 85 86 & $A_1$\\
    33 & 7 & 80 81 82 84 85 86 87 & $A_2$\\
    34 & 8 & 80 81 82 83 84 85 86 87 & $A_1$\\
    \bottomrule
  \end{tabular}
  \caption{The sixteen immediate binary rejections. The indicated four
  selected points span a space of capacity three.}\label{tab:rejected}
\end{table}

No search reached a stopping limit. Every job in the main binary
enumeration had a node budget of $10^8$, above the total of
Table~\ref{tab:outer}, and no profile was found, so the solution limit
never applied. The main ternary run took $7{,}024$ nodes against a budget
of $10^9$ and returned three profiles, one per symmetry orbit, against a
limit of $1{,}000$. The symmetry-free cross-checks described in
Section~\ref{sec:verification} also completed without reaching a stopping limit.

For the ternary output, normalize each projective matrix so its first
nonzero coefficient in row-major order is one, and encode it as
$\sum_{i=0}^1\sum_{j=0}^2 U_{ij}3^{3i+j}$. The common thirteen codes
are
\[
  1,3,9,13,27,28,55,81,84,165,243,252,351.
\]
The extra codes are $364$ for $P_1$, $495$ for $P_2$, and $715$ for
$P_3$. These are precisely the profiles displayed in
Section~\ref{sec:ternary}.

\FloatBarrier

\end{document}